\documentclass[conference]{IEEEtran}
\IEEEoverridecommandlockouts
\usepackage{cite}
\usepackage{amsmath,amsthm,amssymb,amsfonts}
\usepackage{algorithmic}
\usepackage{algorithm}
\newtheorem{theorem}{Theorem}

\usepackage{graphicx}
\usepackage{textcomp}
\usepackage{xcolor}
\usepackage{diagbox}
\def\BibTeX{{\rm B\kern-.05em{\sc i\kern-.025em b}\kern-.08em
    T\kern-.1667em\lower.7ex\hbox{E}\kern-.125emX}}
\def\@IEEEptsizeten{10}
\begin{document}

\title{ Dynamic Quantum Group Key Agreement via Tree Key Graphs\\
{\footnotesize \textsuperscript{}}
\thanks{}
}

\author{\IEEEauthorblockN{1\textsuperscript{st} Qiang Zhao}
\IEEEauthorblockA{\textit{Computer Science and Engineering} \\
\textit{Chinese University of Hong Kong}\\
Hong Kong, China \\
qiangzhao@cuhk.edu.hk}
\and
\IEEEauthorblockN{2\textsuperscript{nd} Zhuohua Li}
\IEEEauthorblockA{\textit{Computer Science and Engineering} \\
\textit{Chinese University of Hong Kong}\\
Hong Kong, China \\
zhli@cse.cuhk.edu.hk}
\and
\IEEEauthorblockN{3\textsuperscript{rd} John C.S. Lui}
\IEEEauthorblockA{\textit{Computer Science and Engineering} \\
\textit{Chinese University of Hong Kong}\\
Hong Kong, China \\
cslui@cse.cuhk.edu.hk}
}

\maketitle
\thispagestyle{plain}
\begin{abstract}
Quantum key distribution (QKD) protocols are essential to guarantee information-theoretic security in quantum communication. Although there was some previous work on quantum group key distribution, they still face many challenges under a ``\textit{dynamic}'' group communication scenario. In particular, when the group keys need to be updated in real-time for each user joining or leaving to ensure secure communication properties, i.e., forward confidentiality and backward confidentiality. However, current protocols require a large amount of quantum resources to update the group keys, and this makes them impractical for handling large and dynamic communication groups. In this paper, we apply the notion of ``{\em tree key graph}'' to the quantum key agreement and propose two dynamic Quantum Group Key Agreement (QGKA) protocols for a join or leave request in group communications. In addition, we analyze the quantum resource consumption of our proposed protocols. The number of qubits required per join or leave only increases logarithmically with the group size. As a result, our proposed protocols are more practical and scalable for large and dynamic quantum group communications.
\end{abstract}

\section{Introduction}
Traditionally, most network applications are based on unicast or multicast to transfer information (i.e., packets). For a large group communication session, the unicast or multicast servers are often required to send packets to a large number of authorized clients within a communication group. In classical networks, multicast group communication has been successfully used to provide an efficient delivery service to a large communication group ~\cite{wong2000graph}. Unfortunately, classical key distribution (e.g., Diffie-Hellman key exchange) can easily be broken by quantum computing. To avoid such quantum attacks, quantum group key distribution aspires to enable efficient and information-theoretical secure communication in a large-scale communication group.

For quantum group communications, one cannot simply use multicast to distribute group keys like in classical communication. This is because quantum states cannot be copied according to the no-cloning theorem~\cite{wootters1982single} in quantum dynamics. Although there are some multi-party quantum group key distribution protocols~\cite{epping2017multi,cao2022evolution,bian2023high}, they often consume large amounts of quantum resources to distribute group keys. Hence, they limit the scalability of the group size. Moreover, developers foresee that there will be many services that require efficient quantum key distribution, e.g. such as conferences, information exchange, and distributed computing. As a result, the quantum group key distribution remains a critical study to ensure the confidentiality, authenticity, and integrity of messages delivered between legitimate members in a large and dynamic quantum communication group. 

Quantum Group Key Distribution (QGKA) is an extension protocol of Quantum Key Distribution (QKD) in group sessions. The Quantum key agreement (QKA) was proposed by Zhou, Zeng, and Xiong~\cite{zhou2004quantum}, which is used to generate a shared key between ``\textit{two}'' legitimate parties. Since then, various QKA protocols for a limited number of users have been developed, and they utilize various quantum resources or quantum techniques to generate group keys, such as Bell states~\cite{shukla2014protocols}, single-photon states~\cite{liu2013multiparty}, cluster states~\cite{liu2019multiparty} and GHZ states~\cite{zeng2016multiparty}, and so on. Group key from one or more servers needs to be distributed to many authorized users in a quantum network. However, for many of these protocols, the number of required qubits increases ``\textit{linearly}'' with the group size, which implies that it is difficult to apply them to large-group communications. In particular, when the communication group is {\em dynamic}, i.e., when users frequently join or leave the session, it is necessary to ensure \textit{ backward and forward confidentiality} in the communications, meaning that the newly joining user cannot access the group communications before she joins the group (backward confidentiality), while the user who leaves cannot access any group communication upon her departure (forward confidentiality). To provide these backward/forward confidentiality properties, one needs to update the group keys in real time whenever there is any membership change in the group communication. However, existing QKA protocols entail a substantial quantum resource overhead for each group key update. 

In group communications, group key management usually depends on the key graph~\cite{wong2000secure}, which is a directed acyclic graph with the nodes of users and keys. In a key graph, the key server knows the keys of all users, and the key node is shared by her child users. The quantum group communications are usually based on a star key graph, i.e., the key server is directly connected to the key node of each user in the group. When a user leaves a group session with $N$ users, if a two-party QKA with Bell states~\cite{shi2013multi,shukla2014protocols} is used to distribute the $1$-bit group key, each user needs to prepare a Bell state. One of the entangled particles needs to be transmitted to other participants in turn, undergoing encryption in the process. The group key is subsequently extracted through Bell measurement. Therefore, the remaining $N-1$ users will consume $2(N-1)$ entangled qubits for a $1$-bit group key. If a multi-party QKA with GHZ states is used, the server needs to create $(N-1)$-qubit entangled state to perform a multi-party QKA. In both cases, a large number of entangled qubits are consumed whenever a user joins or leaves the communication group. For a two-party QKA, Bell states are easy to prepare, but the protocol needs to be performed $N-1$ times for existing users in the group. In terms of multi-party QKA, although the protocol only needs to be performed once, the preparation of the $(N-1)$-qubit GHZ state is more difficult compared with the preparation of the Bell state, particularly for large $N$. For each group key update, a substantial quantum resource overhead is required to perform preparation, transmissions, measurements, and quantum operations. The above reasons are why these current protocols severely hinder the development of large dynamic quantum applications.

In quantum group communication, the problem of how to distribute a secure key to a group of users has been addressed in the quantum cryptography literature~\cite{epping2017multi,chou2018dynamic,liu2019multiparty,xu2020secure,bian2023high}. However, no one has addressed the associated frequent group key changes and the scalability problem for a large and dynamic group. In classical communications, the hierarchy approach of keys~\cite{wallner1999key,wong2000graph}, called the tree key graph or key tree, has been applied in group key distributions to improve scalability and reduce the number of encryptions. To support a large-scale and dynamic quantum group communication, we introduce the hierarchy approach to quantum key distribution protocol and propose the ``{\em key tree-based}'' QGKA protocol for a join or leave to reduce quantum resource consumption. Our approaches not only can guide the QKA protocol to update the keys in the key tree, but can also make full use of the existing keys to build the rekey messages. Our main contributions to this paper are as follows:

\begin{itemize}
    \item[$\bullet$] We propose two dynamic QGKA protocols for a join or leave event to redistribute secure group keys. Existing group users are only required to update the keys associated with a joining or leaving user in the key tree, which can significantly reduce the quantum resource consumption.
    \item[$\bullet$] We present the security of our algorithms by examining both generation and distribution processes in the key tree framework. Our analysis shows that the generation is secure against external and internal quantum attacks, and the distribution designed based on the key tree can be used to enhance its resistance against quantum attacks. Consequently, both reinforce the overall security of dynamic group communications.
    \item[$\bullet$] We present the consumption of quantum resources in our algorithms and compare them with some well-known QKA protocols. We show that the number of qubits required per join or leave only increases ``\textit{logarithmically}'' with group size. As a result, our protocols are scalable to large groups with frequent joins and leaves.  
\end{itemize} 
\section{Related works}
In quantum communication, quantum key distribution~\cite{sasaki2011field,dynes2019cambrige,chen2021network} has attracted much attention since it can realize information-theoretical security that cannot be achieved by classical key distribution. For a communication group, if each key is randomly generated by the key server, then the key server has full control over all group keys. However, when the server is untrusted or malicious, group keys face the risk of compromise in group communications. To address this issue, QKA~\cite{zhou2004quantum,zeng2016multiparty} is introduced in quantum group communications. In this framework, the group key is generated collaboratively by all participants and cannot be controlled by any single one. 

QKA is a high-security extension protocol of QKD~\cite{bennett2014quantum,ekert1991quantum,bennett1992quantum}. In QKA based on Bell states~\cite{shi2013multi,shukla2014protocols}, users utilize quantum operations and Bell measurements to extract the agreed key. Nonetheless, this approach encounters challenges in terms of scalability and efficiency when applied to larger groups. To facilitate the implementation of QKA, Huang et al.~\cite{shukla2014protocols} use single photons to realize QKA based on the travelling mode. However, efficiency levels remain modest. Sun et al.~\cite{sun2016efficient} sought to optimize key efficiency by introducing dense coding in their proposal for efficient multi-party QKA with four-qubit cluster states. Despite these innovations, quantum transmission mainly relies on unicast in their protocols. To accommodate multicast communications, QKA based on GHZ states~\cite{xu2014novel,zeng2016multiparty} has emerged. It allows the server to simultaneously transmit GHZ states to multiple users, empowering them to engage in quantum operations and GHZ measurements collaboratively for group key agreement, constituting a multicast distribution mechanism. However, the practical realization of multi-qubit GHZ states for a large group faces substantial challenges.   

For a communication group, when a user joins (or leaves) a group session, the communication group key must be updated to ensure that the user cannot obtain previous (or subsequent) communications. For current large quantum communication networks~\cite{cao2022evolution,bian2023high}, the update of group keys will incur a substantial overhead of qubits in preparation, transmission, operation, and measurement, which seriously limits the overall scalability of quantum group communications. Consequently, it remains a challenge to realize dynamic QGKA with frequent joins or leaves in large-scale communication networks.  

In classical networks, the tree key graph as a logical hierarchy of keys is used to facilitate group key distribution among users~\cite{wong2000graph,mittra1997multicasting}. For quantum networks, we introduce a key tree-based hierarchical structure to the QGKA so to update the group keys when users join or leave the group. This approach helps us to realize a secure and efficient dynamic quantum group communication.

\section{Background}

Since the previous QKA protocols\cite{shukla2014protocols,huang2016improved,sun2016efficient,zeng2016multiparty,chou2018dynamic} consume a large number of quantum resources to update the group key when a user joins or leaves a communication group, this will limit the scalability of the group and affect the quantum communication performance. Hence, we propose to use the group key management strategy on a key tree to guide the QKA to update the group keys. We will later show that this will provide a scalable feature for quantum group communication. In what follows, we first provide some background on the quantum key agreement (QKA), then discuss the concept of ``{\em tree key graph}''. 

\subsection{Multi-party QKA}
Quantum key agreement (QKA) is used for all participating users to agree on a single group key for private group communication.
To prevent a malicious server or user from determining the group key,
QKA must satisfy two essential requirements: 
(1) The group key has to be generated by {\em all} users. Each user can also request to {\em change/update} the group key;
(2) only users in a communication group can obtain the group key, users outside the group can not obtain the group key.
\begin{figure}[htbp]
	\centering
	\includegraphics[width=3in]{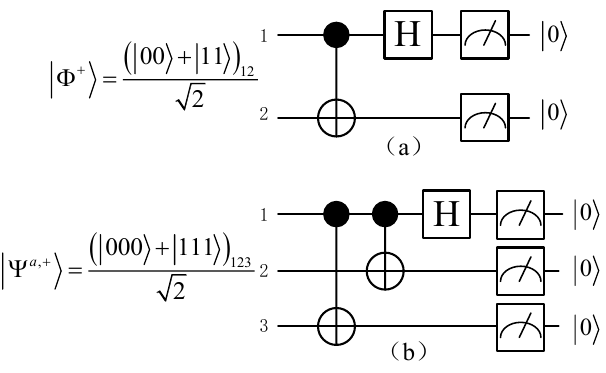}
	\caption{Entanglement measurement. (a) Bell measurement. (b) GHZ measurement. When the initial quantum state is $|\Phi^+\rangle$ or $|\Psi^{a,+}\rangle$, the quantum state after measurement is $|00\rangle$ or $|000\rangle$, respectively.}
\label{fig-measurement}
\end{figure}
In QKA, the quantum measurement is a critical step in extracting the shared key. For two parties, the measurements of two-qubit entangled states are often called EPR or Bell measurements. As shown in Fig.\ref{fig-measurement}(a), there can be four measurement results:
\begin{equation}\label{eq-bell-measure}
\begin{aligned}
{\left| {{\Phi ^ \pm }} \right\rangle _{12}} =&\frac{{{{\left( {\left| {00} \right\rangle  \pm \left| {11} \right\rangle } \right)}_{12}}}}{{\sqrt 2 }}\stackrel{M}\longrightarrow|00\rangle_{12},|10\rangle_{12},\\
{\left| {{\Psi ^ \pm }} \right\rangle _{12}} =&\frac{{{{\left( {\left| {01} \right\rangle  \pm \left| {10} \right\rangle } \right)}_{12}}}}{{\sqrt 2 }}\stackrel{M}\longrightarrow|01\rangle_{12},|11\rangle_{12},\\
\end{aligned}
\end{equation}
where the subscript indexes indicate the position of the first and second qubits. For multiple parties, the measurements of three or more entangled qubits are called GHZ measurements. As shown in Fig.\ref{fig-measurement}(b), the three-qubit GHZ measurements generate eight results:
\begin{equation}\label{eq-ghz-mesure}
\begin{aligned}
{\left| {{\Psi ^{a, \pm }}} \right\rangle _{123}} =&\frac{{{{\left( {\left| {000} \right\rangle  \pm \left| {111} \right\rangle } \right)}_{123}}}}{{\sqrt 2 }}\stackrel{M}\longrightarrow|000\rangle_{123},|100\rangle_{123},\\
{\left| {{\Psi ^{b, \pm }}} \right\rangle _{123}} =&\frac{{{{\left( {\left| {001} \right\rangle  \pm \left| {110} \right\rangle } \right)}_{123}}}}{{\sqrt 2 }}\stackrel{M}\longrightarrow|001\rangle_{123},|101\rangle_{123},\\
{\left| {{\Psi ^{c, \pm }}} \right\rangle _{123}} =&\frac{{{{\left( {\left| {010} \right\rangle  \pm \left| {101} \right\rangle } \right)}_{123}}}}{{\sqrt 2 }}\stackrel{M}\longrightarrow|010\rangle_{123},|110\rangle_{123},\\
{\left| {{\Psi ^{d, \pm }}} \right\rangle _{123}} =&\frac{{{{\left( {\left| {011} \right\rangle  \pm \left| {100} \right\rangle } \right)}_{123}}}}{{\sqrt 2 }}\stackrel{M}\longrightarrow|011\rangle_{123},|111\rangle_{123},
\end{aligned}
\end{equation}
where the superscript indexes indicate the different entangled states, and the subscript is the position of the qubit in the GHZ state. Note that in Eq \eqref{eq-bell-measure} and \eqref{eq-ghz-mesure}, the different entangled states have different measurement results, so one can deduce the measured entanglement states based on Bell or GHZ measurement. If the entangled state is changed, one can determine the quantum operations on qubits. For example, if the initial quantum state is $|\Phi^+\rangle_{12}$ or $|\Psi^{a,+}\rangle_{123}$, when the state is acted by the unknown quantum operation, the Bell or GHZ measurement result is $|01\rangle$ or $|011\rangle$, so one can deduce that the measured state is $|\Psi^+\rangle_{12}$ or $|\Psi^{d,+}\rangle_{123}$, thus one can deduce that the operation gate $X$ acted on the first qubit. Therefore, when one party uses quantum operation to act on Bell or GHZ states, others can extract the key (or quantum operation) through Bell or GHZ measurement to generate a shared group key.  

For a group session with multiple parties, one can utilize the multi-party QKA to agree on the group key. A typical multicast QKA, which was proposed by Zeng et.al~\cite{zeng2016multiparty}, as shown in Fig.\ref{fig-MQKA-diagram}. The server $s$ prepares $N$-qubit GHZ states and shares them with other $N-1$ users. Formally, the GHZ state is:
\begin{equation}\label{eq-GHZ-state}
    \left| {{\rm{GHZ}}} \right\rangle {\rm{ = }}\frac{1}{{\sqrt 2 }}\left( {{{\left| 0 \right\rangle }^{ \otimes N}} + {{\left| 1 \right\rangle }^{ \otimes N}}} \right),
\end{equation}
where $N$ is the number of users and servers in the quantum communication group. This protocol contains three steps:
\begin{itemize}
\item The server first prepares a series of GHZ states of $N$ particles, then randomly inserts decoy states into the sequences, and sends them to each user in the group. 

\item Each user can use the decoy states to detect eavesdroppers, if any. If the channel is secure, the user applies Pauli gates on qubits in his sequence based on his operation keys. To prevent the group key from being controlled by a single user or server, each GHZ state will be assigned randomly to one user or server as a leader, while the others will act as followers. Each user then inserts his decoy states in his sequence and sends it to the leader.

\item The leader of each GHZ state performs GHZ measurements and announces the results. Each member can extract the operation keys of all members and perform XOR operations between them to generate the group key.
\end{itemize}
\begin{figure}[htbp]
	\centering
	\includegraphics[width=3in]{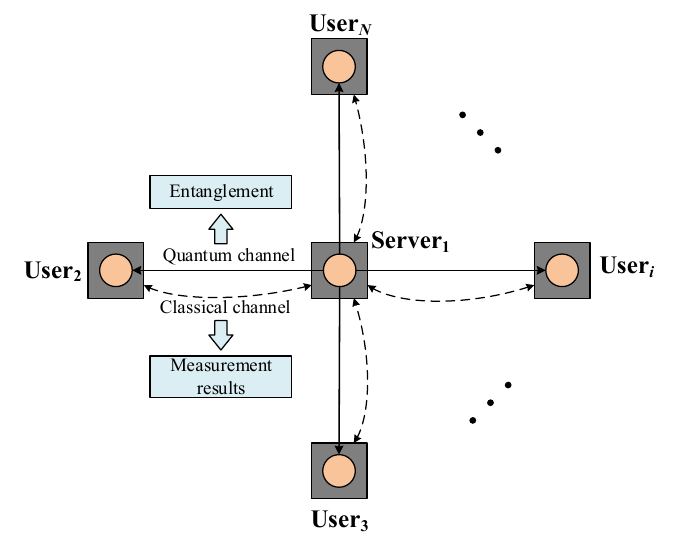}
	\caption{Schematic diagram of $N$-party QKA to generate a group key. Circles are qubits, and the solid black lines connecting them represent quantum channels to distribute GHZ entangled states. The dashed lines between the Server and the Users represent classical channels that are used to transmit measurement results.}
\label{fig-MQKA-diagram}
\end{figure}  

When a user joins or leaves a group session, the remaining users need to generate a new group key to ensure backward or forward confidentiality in the group communications. Based on the above multi-party QKA, a large number of entangled qubits and measurements will be required to update group keys for large groups (i.e., for a group of size $N=10,000$, we need around $10,000$ entangled qubits for GHZ state, and $10,000$ measurements for GHZ state in Fig.\ref{fig-measurement}(b)).

\subsection{Tree Key Graph}
For a quantum group communication session, secure communication between users in the group depends on the shared group key, which can be generated by multi-party QKA. In dynamic group quantum communication, to ensure forward or backward confidentiality, the group key needs to be updated for a join or leave. However, as we stated previously, a new group key distribution based on previously proposed QKD or QKA protocols will consume significant amount of quantum resources, which is inefficient for large groups. In classical communications, a tree key graph is a virtual hierarchy topology instead of a logical star key graph in group key management~\cite{wong2000graph}. This tree key graph can be used to reduce the number of encryptions in communications. Let us briefly present the key-tree concept here.

\begin{figure}[htbp]
	\centering
	\includegraphics[width=3in]{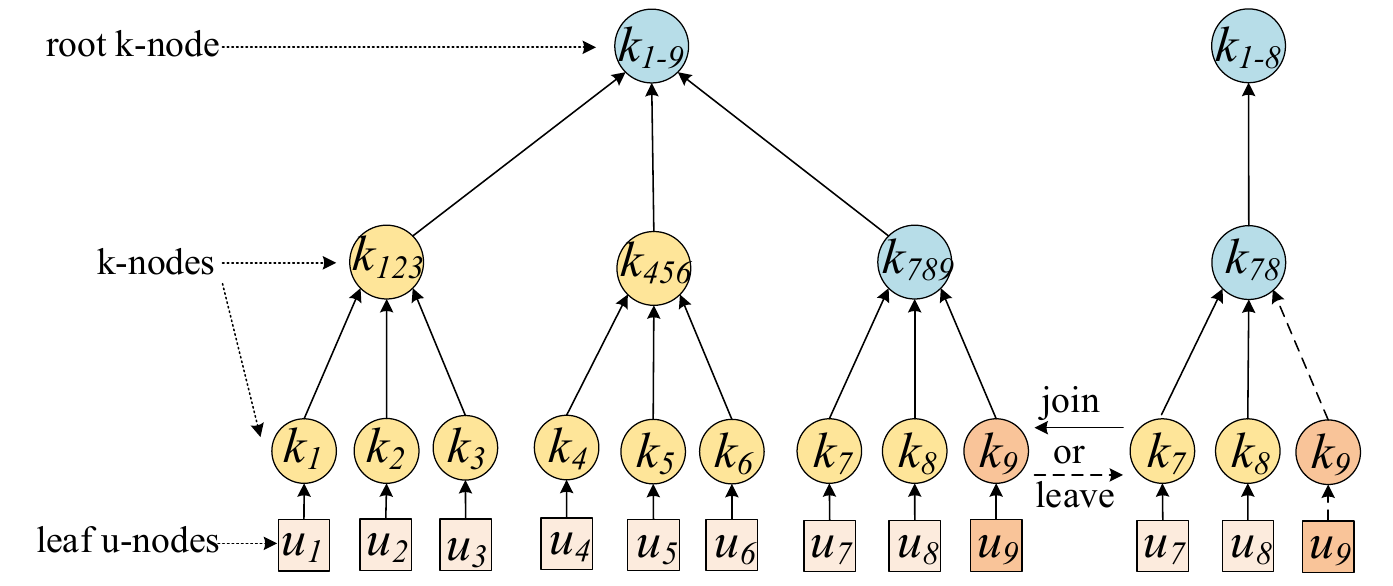}
	\caption{A tree key graph with nine users.}
\label{fig-nine-key-tree}
\end{figure}

A key tree is a directed acyclic graph $G$ as depicted in Fig.\ref{fig-nine-key-tree}, which consists of user nodes and key nodes, denoted as "{\em $u$-nodes }" and "{\em $k$-nodes}" respectively. Each $u$-node only has an outgoing edge, while each $k$-node has one or more incoming edges. If a $k$-node only has incoming edges but no outgoing edges, then we call it a root node. A tree key graph $G$ is a single-root tree. To accurately describe the tree key graph, one needs to introduce two parameters, namely the height $h$ and the degree $d$. The height $h$ is defined as the number of edges along the longest directed path from the $u$-node to the root $k$-node in the tree $G$. The degree $d$ is defined as the maximum number of a $k$-node's incoming edges in the tree $G$. For example, in the tree key graph in Fig.\ref{fig-nine-key-tree}, we have $h=3$ and $d=3$. Note that when $N$ is group size, we have $N=d^{h-1}$, and the number of keys is $(d^h-1)/(d-1)\approx (d/(d-1))N$ in a full and balanced tree key graph.

In a tree key graph, each user possesses ``{\em multiple keys}'', which are keys along the path from her $u$-node to the root $k$-node, so a user has at most $h$ keys. The key at the root of the key tree is used as the group key, and it is shared by all users. For example, we present a tree key graph with nine users in Fig.\ref{fig-nine-key-tree}. To illustrate the relationship of users and keys, the set of keys held by the user $u$ is represented by the associated key function $\textit{keyset(u)}$, i.e., $keyset(u_9)=\{k_9,k_{789},k_{1-9}\}$. The set of users that hold the key $k$ is represented by the associated user function $\textit{userset(k)}$, i.e., $userset(k_{789})=\{u_7,u_8,u_9\}$.  

As shown in Fig.\ref{fig-nine-key-tree}, when the user $u_9$ joins the group, to ensure that she cannot obtain the previous communications, the server needs to update the keys $\{k_{78},k_{1-8}\}$ if $u_9$ is connected to the joining node $k_{78}$. To achieve this, the server encrypts the new key $k_{789}$ ($k_{1-9}$) with $k_{78}$ ($k_{1-8}$) based on some classical encryption algorithms, and redistributes these keys to $\{u_7,u_8\}$ (the users in ${userset(k_{1-8})}$), and uses $k_9$ to encrypt $k_{789}$ and $k_{1-9}$ to send to $u_9$. In this process, the server performs $4$ encryptions instead of $2$ in the star key graph. When a user $u_9$ leaves the group, to ensure she cannot obtain any subsequent communications, the server needs to update keys $\{k_{789},k_{1-9}\}$. To achieve this, the server encrypts the new key $k_{78}$ with $k_7$ ($k_8$) so to securely send it to $u_7$ ($u_8$), and encrypts the new group key $k_{1-8}$ with the subgroup keys $k_{78}$, $k_{123}$ and $k_{456}$, and then sends to users in ${userset(k_{1-9})}-u_9$. In the process, the server performs only $5$ encryptions instead of $8$ in the star key graph. Therefore, the average cost of encryptions for this case is $4.5$ for a join and leave. For a large group size of $N$, the server needs to perform $O(\log N)$ encryptions in the tree key graph rather than $O(N)$ in conventional a single-level connection topology. 

For a large communication group in a quantum network, when a user joins (or leaves), we need to ensure that she cannot access previous (or subsequent) communications. Hence, the group key shared by the online group members must be updated. In QKA protocols with Bell or GHZ, the consumption of quantum resources, i.e., the number of qubits, for each key update is $O(N)$, which implies that it is a technical challenge to address large-scale dynamic quantum group communications.
 
The tree key graph as a hierarchical key approach can be used to assist servers and users in distributing the group keys. Based on the core idea of the key tree for classical networks, it can also be generalized to quantum group communication with frequent joins or leaves. Hence, we will propose the key tree-based QGKA protocol for a join or leave request to instruct the key server to redistribute group keys, enabling seamless and robust communications within dynamic groups.

\section{Dynamic quantum group key agreement}
In a quantum network, the quantum key distribution as a typical distribution protocol has been applied to number of fields~\cite{liu2022towards,cao2022evolution}. 
The current network architectures have different layers based on definitions and applications~\cite{masahide2011field,cao2019kaaskey,chen2021integrated,ITUT2019overview}. One of the most typical quantum network architecture is a three-layer architecture of the quantum key distribution network. This architecture includes three logical layers: (1) the infrastructure layer, (2) the key management layer, and (3) the application layer. Note that the key management layer is a bridge connecting the infrastructure layer and the application layer. As shown in Fig.\ref{fig-Three-layer-quantum-network}, this layer not only can instruct the quantum nodes of the infrastructure layer to perform QKA between nodes, but also can provide various services for the application layer. To address dynamic quantum group communications, our approach contains two processes: the generation and distribution of new keys. The generation process aims to reduce quantum resource consumption by leveraging QKA to update the keys associated with the joining or leaving user in the key tree. The distribution process is designed to take full advantage of the keys unknown to the joining or leaving user in the key tree so to construct the rekey messages. 
\begin{figure}[htbp]
	\centering
	\includegraphics[width=2in]{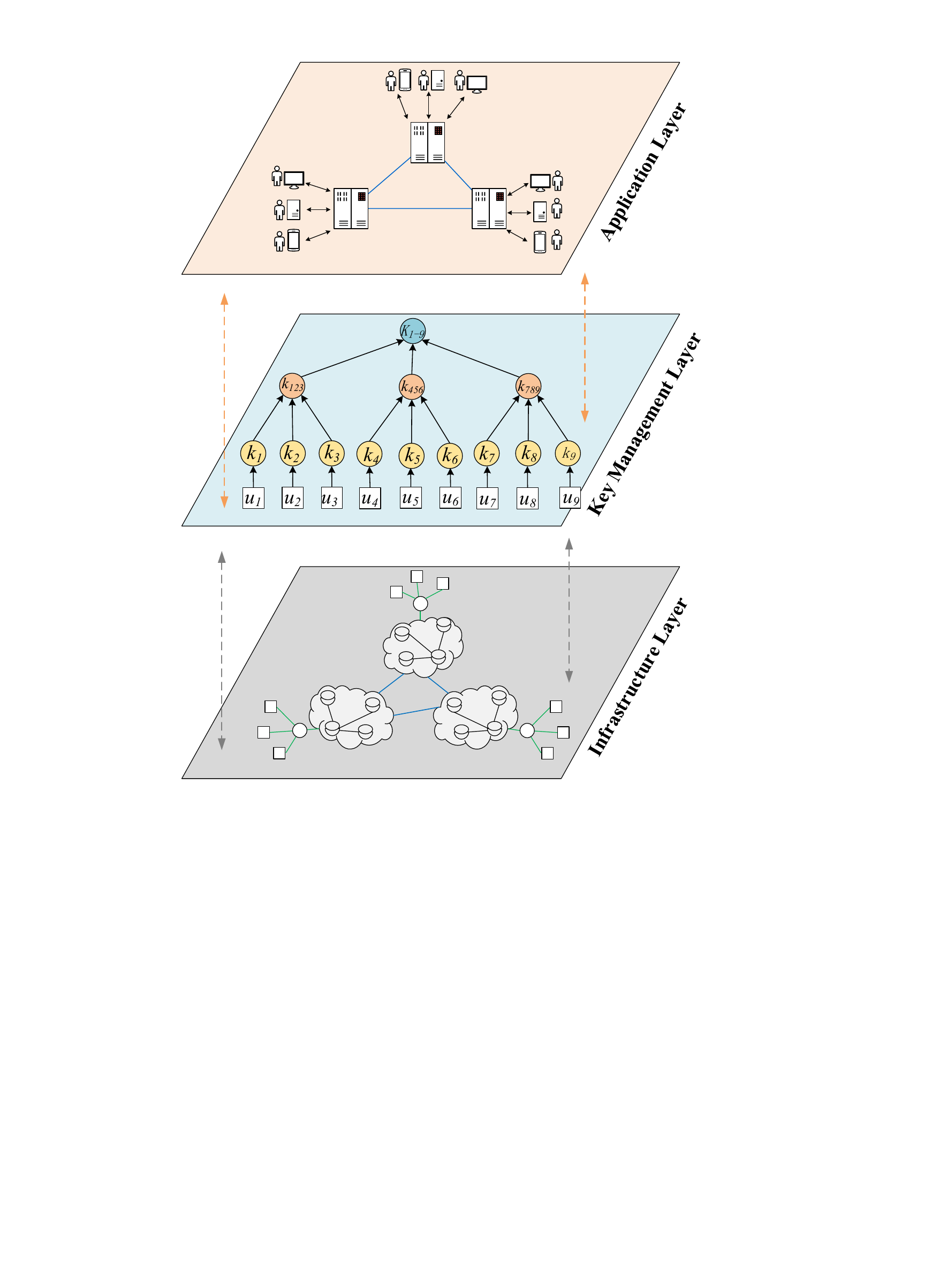}
	\caption{Three-layer architecture quantum network.}
\label{fig-Three-layer-quantum-network}
\end{figure}

\subsection{QGKA with a join}

\begin{algorithm}
     \caption{QGKA protocol for a join request}
     \label{alg-qka-join}
     \begin{algorithmic}[1]
         \STATE The server $s$ grants the user $u$ a join request.\\
         \textbf{Generation:}
         \FOR{$i=0$ to $h-2$}  
         \STATE Let $k_{h-2}$ be the parent k-node of $u$, and $k_{i}$ be the parent k-node of $k_{i+1}$.        
         \STATE The server $s$ prepares the Bell state $|\psi\rangle_{i}$, randomly inserts the decoy states in one of the Bell sequences, and sends it to user $u$.           
         \STATE The server $s$ publishes the bases and positions of the decoy states for channel checking. If the error rate is not higher than the threshold, they continue.        
         \STATE The server $s$ is a leader at all odd positions in the Bell sequence; otherwise, she is a follower. The leader performs one of the four Pauli operations $\{I, X, Y, Z\}$, and the followers perform one of $I$ and $X$. The follower then randomly inserts decoy states in the Bell sequence she holds and sends it to the leader.         
         \STATE After channel checking, the leader performs Bell measurements and publishes the measurement results. They can exact the operation keys of the other and generate the new key $k_{i}^{\prime}=\kappa_i^s\oplus \kappa_i^u$.
         \ENDFOR \\
         \textbf{Distribution:}
         \FOR{$j=0$ to $h-2$}
         \STATE Let $KM_j=\{\{k_0^{\prime}\}_{k_0},...,\{k_j^{\prime}\}_{k_j}\}$, the server $s$ sends $KM_j$ to the users of $userset(k_j)-userset(k_{j+1})$ in Fig.\ref{fig-N-key-tree}. 
         \ENDFOR           
     \end{algorithmic}
 \end{algorithm}

 When a user $u$ joins, we first locate her joining point $k$-node. In each subgroup of the key tree, the parent node of the user from the smallest subgroup can be used as the joining point. To ensure that she cannot access the previous information (e.g., backward confidentiality), we need to update those keys along the path from her parent $k$-node to the root $k$-node, as we have shown in Fig.\ref{fig-N-key-tree}. In the generation process, let $k_{i-1}$ be the parent $k$-node of $k_i, i=1,2,...,h-1$, and the root $k$-node $k_0$ is the group key. To deal with the malicious attacks, the new keys $k_i^{\prime}$ are generated jointly by the server $s$ and the joining user $u$ based on a two-party QKA protocol. In the distribution process, the old keys $k$ are unknown to the joining user $u$, i.e., all yellow nodes and some blue nodes in Fig.\ref{fig-N-key-tree}. The server $s$ can use these keys as seed keys in the Advanced Encryption Standard (AES)-256, and refresh them every second~\cite{pan2018satellite}. Therefore, the new group keys $k_i^{\prime}$ can be encrypted by $k_i$ based on AES-256 to construct the rekey messages $KM_i=\{k_i^{\prime}\}_{k_i}$. Finally, they are redistributed to group users based on the key tree in Fig.\ref{fig-N-key-tree}. In this work, we propose the QGKA protocol for a join request so to realize secure and efficient group communications. This algorithm is shown in Algorithm \ref{alg-qka-join}.

\begin{figure}
    \centering
    \includegraphics[width=3in]{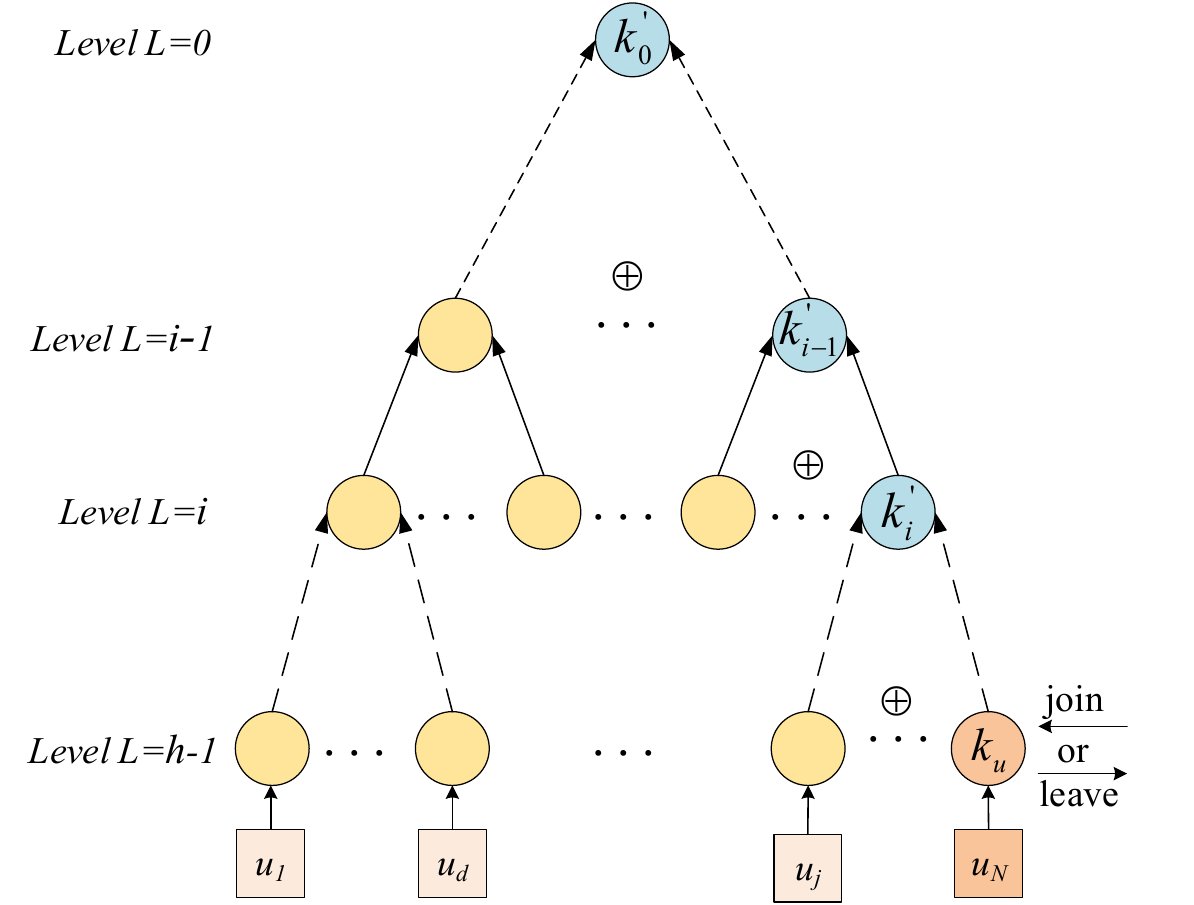}
    \caption{Schematic diagram of the tree key graph $G$ with $N$ users. The $h$ is the height of $G$, and the $d$ is the degree of $G$. The keys of the blue nodes need to be updated for a join or leave request.}
    \label{fig-N-key-tree}
\end{figure}

\begin{table}[htbp]
    \centering
    \caption{Key generation of two-party QKA}
    \label{table-two-key-generation}
    \begin{tabular}{|c|c|c|c|c|}
         \hline        \diagbox{Follower}{Key}{Leader} &$I_1^{L}\rightarrow0$ & $X_1^L\rightarrow0$ & $Y_1^L\rightarrow1$ & $Z_1^{L}\rightarrow1$\\
         \hline        $I_2^F\rightarrow0$ & $|00\rangle=0$ & $|01\rangle=0$ & $|11\rangle=1$ & $|10\rangle=1$ \\       
         \hline         $X_2^F\rightarrow1$ & $|01\rangle=1$ & $|00\rangle=1$ & $|10\rangle=0$ & $|11\rangle=0$\\
         \hline
    \end{tabular}
\end{table}

In Algorithm \ref{alg-qka-join}, the server and the joining user can obtain each other's operation key based on the Bell state measurement to extract the group keys. Based on the original QKA protocol~\cite{zeng2016multiparty}, the leader uses quantum gates $I$ and $X$ to represent the operation key $0$, and quantum gates $Y$ and $Z$ to represent the operation key $1$; the follower uses quantum gates $I$ and $X$ to represent his operation keys $0$ and $1$, respectively. The leader performs Bell measurements and sends the measurement results to the follower. Then, they can generate the shared key $k_{i}^{\prime}=\kappa_i^s\oplus \kappa_i^u$ based on Table \ref{table-two-key-generation}. For example, if the Bell measurement result is $|01\rangle$, the leader can infer the follower's operation $I$ through his operation $X$, and so can the follower. Therefore, the shared key is $0\oplus0=0$ as shown in Table \ref{table-two-key-generation}.

For example, if the user $u_9$ is the joining user in Fig.\ref{fig-nine-key-tree}, the server needs to jointly generate new group keys with the user $u_9$, and then sends these three rekey messages to other users as:
\begin{equation}\label{eq-join-messages}
\begin{aligned}
s &\stackrel{Q}\longleftrightarrow u_9 &&:\{k_{1-9},k_{789}\}, \\
s &\longrightarrow \{u_7,u_8\}&&:\{k_{1-9}\}_{k_{1-8}},\{k_{789}\}_{k_{78}},\\
s &\longrightarrow \{u_1,u_2,...,u_6\}&&:\{k_{1-9}\}_{k_{1-8}},\\
\end{aligned}
\end{equation}
where the first formula in Eq \eqref{eq-join-messages} represents that the server $s$ and the joining user $u_9$ perform QKA to generate group keys $k_{1-9},k_{789}$, and the number of qubits needed to generate a 1-bit group key is $4$. The second formula indicates that the server $s$ sends the rekey messages $\{k_{1-9}\}_{k_{1-8}}, \{k_{789}\}_{k_{78}}$ to the users $u_7,u_8$. The $\{k_{1-9}\}_{k_{1-8}}$ represents that the server uses the old key $k_{1-8}$ to encrypt the new key $k_{1-9}$ based on symmetric encryption algorithms, such as Advanced Encryption Standard (AES)-256. 

In Algorithm \ref{alg-qka-join}, each new key $k_i^{\prime}$ needs to be generated and distributed to other group users. The new joining user can participate in the generation process of these new keys based on a two-party QKA protocol, so quantum resource consumption of $2(h-1)$ if the decoy states for checking the channel are ignored. In the distribution process, each new key needs to be encrypted once to build the rekey messages, so the encryption cost is $h-1$. Except for the joining user, the rekey messages for each user can be combined into one message, so the number of rekey messages for the new keys is $h-1$, as shown in Eq \eqref{eq-join-messages}. More importantly, the QKA in our algorithm can ensure information-theoretic security of the new keys generated in quantum communications. 

\subsection{QGKA with a leave}

When a user $u$ leaves a group session, the server $s$ updates the key tree by deleting the $k$-node and the $u$-node of the user $u$. To ensure that the leaving user $u$ cannot obtain subsequent communications, those keys $k_i$ associated with $u$ must be updated, i.e., the blue nodes along the path from $k_u$ to the root node $k_0$ in Fig.\ref{fig-N-key-tree}. In the generation process, for the key $k_i$, the server $s$ and the users in $userset(k_i)$ except for $u$ first utilize multi-party QKA to generate the new key $k_i^{\prime}$. Note that each new key is jointly generated by all participants, which can prevent the key from being controlled by a single user or server. In the distribution process, a child of the key $k_i$ is denoted as $k_i^c,c=1,2,...,d$. The server $s$ then uses $k_i^c$ to encrypt the new key $k_i^{\prime}$ to build the rekey messages $\{k_i^{\prime}\}_{k_i^c}$, and redistribute it to other users in $userset(k_i^c)$. Consequently, we propose the QGKA protocol for a leave request, and it is depicted in Algorithm \ref{alg-qka-leave}.

\begin{algorithm}
     \caption{QGKA protocol for a leave request}
     \label{alg-qka-leave}
     \begin{algorithmic}[1]
         \STATE The server $s$ grants the user $u$ a leave request.\\
         \textbf{Generation: }      
         \FOR{$i=0$ to $h-2$}
         \FOR{$c=1$ to $d$}
         \STATE The $k_i^c$ is the child key of $k_i$, In $userset(k_i^c)$, the users randomly select a user as their agent $u_i^c (\neq u)$.
         \ENDFOR
         \STATE The server $s$ prepares a series of $(d+1)$-qubit GHZ states $|GHZ\rangle_{i}$ and shares them with these agent users $\{u_i^c\}_{c=1}^d$. 
         \STATE Similarly to Algorithm \ref{alg-qka-join}, the server $s$ and the agent users $u_i^c$ alternate as a leader for each GHZ state, and others are followers. Based on the operation rule in Table \ref{table-three-key-generation}, the leader and the followers perform operations on their own qubits. The followers then return their entangled particles to the leader.      
        \STATE After channel checking, all the leaders perform GHZ measurements and publish the measurement results. Each participant can obtain the operations of others based on the leader's measurement results in Table \ref{table-three-key-generation}, thus generating the new key $k_{i}^{\prime}=\bigoplus_{c=1}^{d}\kappa_{i}^c\oplus \kappa_i^s$, where $\kappa_{i}^c$ and $\kappa_i^s$ are the operation keys of $u_{i}^c$ and the server $s$, respectively.
        \ENDFOR \\
        \textbf{Distribution:}
        \FOR{$j=0$ to $h-2$}
        \STATE Let $k_j^c$ be the child node key of $k_j$. Let $KM_j=\{k_j^{\prime}\}_{k_j^c}$,
        the server $s$ sends $KM_j$ to other users in $userset(k_j^c)\setminus\{u\}$.
        \ENDFOR           
     \end{algorithmic}
 \end{algorithm}

\begin{table}[htbp]
    \centering
    \caption{Key generation of three-party QKA}
    \label{table-three-key-generation}
    \begin{tabular}{|c|c|c|c|c|}
         \hline
         \diagbox{Follower}{Key}{Leader} &$I_1^{L}\rightarrow0$ & $X_1^L\rightarrow1$ & $Y_1^L\rightarrow0$ & $Z_1^{L}\rightarrow1$\\
         \hline         
         $I_2^FI_3^F\rightarrow00$ & $|000\rangle=0$ & $|011\rangle=1$ & $|111\rangle=0$ & $|100\rangle=1$ \\
         \hline
         $I_2^FX_3^F\rightarrow01$ & $|001\rangle=1$ & $|010\rangle=0$ & $|110\rangle=1$ & $|101\rangle=0$\\
         \hline
         $X_2^FI_3^F\rightarrow10$ & $|010\rangle=1$ & $|001\rangle=0$ & $|101\rangle=1$ & $|110\rangle=0$\\
         \hline
         $X_2^FX_3^F\rightarrow11$ & $|011\rangle=0$ & $|000\rangle=1$ & $|100\rangle=0$ & $|111\rangle=1$\\
         \hline
    \end{tabular}
\end{table}

In Algorithm \ref{alg-qka-leave}, note that our operations on GHZ states satisfy the Zeng's rule ~\cite{zeng2016multiparty}. That is, when the number of users and servers is odd, the leader's operations $I, X, Y, Z$ represent $0,1,0,1$, respectively, and when the number of users and servers is even, they represent $0,0,1,1$, respectively. The follower's operations $I$ and $X$ always represent $0$ and $1$, respectively. 

For the leader, when her operation is $X$ or $Y$ for each GHZ state, she performs the NOT gate on the measurement result to obtain the operation key. The follower's result $|0\rangle$ represents that her operation key is $0$, and $|1\rangle$ is $1$. For example, the key generation of a three-party QKA is shown in Table \ref{table-three-key-generation}. If the leader performs $Y_1^L$ on the first qubit, and the measurement result of the leader is $|101\rangle$, then her operation key is $0$ and she performs the NOT gates on the measurement result $|101\rangle$ to get $|010\rangle$. The second and third qubits are the result of the followers $|10\rangle$, which indicates that their operations are $X_2^F$ and $I_3^F$, respectively, i.e., $1$ and $0$. Therefore, the leader can generate the agreement key as $0\oplus1\oplus0=1$.

For the follower, when her measurement result is different from her operation key, she performs the NOT gates on the measurement result of GHZ state to get the operation key of each participant. For example, as shown in Table \ref{table-three-key-generation}, if the measurement result is $|101\rangle$, and the follower performs $X_2^F$ on the second qubit, which means that its operation key is $1$. She can perform the NOT gates on the $|101\rangle$ to get $|010\rangle$, which indicates their operation keys $0$, $1$, and $0$. Consequently, the follower can also exact the agreement key as $0\oplus1\oplus0=1$.

For example, if $u_9$ in Fig.\ref{fig-nine-key-tree} is granted to leave the group session, the associated keys for leaving k-node $k_9$ contain the k-node $k_{789}$ and the group key $k_{1-9}$, which need to be updated as $k_{78}$ and $k_{1-8}$, respectively. The server $s$ and the users $\{u_7,u_8\}$ need to perform three-party QKA to generate the key $k_{78}$, the server $s$ and the users $\{k_1,k_4,k_7\}$ can perform four-party QKA to generate the group key $k_{1-8}$, and the number of qubits needed is $7$. The following gives three rekey messages:
\begin{equation}\label{eq-leave-messages}
\begin{aligned}
s &\stackrel{Q}\longleftrightarrow \{u_1,u_4,u_7\} &&:k_{1-8} \\
s &\longrightarrow \{u_2,u_3\} &&:\{k_{1-8}\}_{k_{123}}\\
s &\longrightarrow \{u_5,u_6\} &&:\{k_{1-8}\}_{k_{456}}\\
s &\stackrel{Q}\longleftrightarrow \{u_7,u_8\} &&:k_{78} \\
s &\longrightarrow \{u_8\} &&:\{k_{1-8}\}_{k_{78}} \\
\end{aligned}
\end{equation}

The leaving protocol illustrates how the group keys are generated and distributed the rekey messages to the remaining users, and this is shown in Algorithm \ref{alg-qka-leave}. Note that the number of associated keys for the user leaving is $h-1$. In each key update, the server and $d$ users in the subgroup are required to participate in multi-party QKA. Therefore, the number of qubits needed to generate a $1$-bit group key is $d(h-1)$ ignoring the decoy states for channel checking. In the distribution process, each new key needs to be encrypted $d$ times except the last keys since no subsequent child k-node keys need to be encrypted for distribution to users. Hence, the encryption cost of the leaving protocol is $d(h-2)$. According to the distribution in Algorithm \ref{alg-qka-leave}, the number of rekey messages is $d(h-2)$. Each new key $k_i^{\prime}$ needs to be generated by multi-party QKA, so the number of QKA executions is $h-1$.

\section{Performance evaluation}
In the realm of dynamic quantum group key distribution, our evaluation focuses on two primary indicators: security and quantum resource consumption. In this section, we carry out a comprehensive analysis of our proposed algorithms, examining their performance through the lens of these critical factors. 
\subsection{Security analysis}

To illustrate the security of Algorithms \ref{alg-qka-join} and \ref{alg-qka-leave}, we consider two aspects, i.e., generation and distribution of the new keys. In the generation process, we need to update the keys associated with the joining or leaving user $u$, which are the blue k-node keys in the key tree in Fig.\ref{fig-N-key-tree}. In Algorithms \ref{alg-qka-join} and \ref{alg-qka-leave}, each new key is jointly generated by the server and the group users based on QKA. As depicted in Tables \ref{table-two-key-generation} and \ref{table-three-key-generation}, even if the attackers get the measurement results, they will be unable to deduce the generated keys. This means that the new keys consistently exhibit unpredictability for potential attackers. 

For external attacks, we can utilize the decoy state technique to prevent them from obtaining key information. For instance, if Eve eavesdrops on the quantum channels between the server and the users, she will introduce errors in the transmitted quantum sequences, which will cause her to be discovered. For example, when the decoy states $\{|0\rangle,|1\rangle,|+\rangle,|-\rangle\}$ are inserted in the transmitted sequences, the error rate introduced by Eve is $1/4$ with a single qubit~\cite{chun2005secure}, similar to that in BB84 QKD~\cite{bennett1984quantum}. Consequently, the probability that Eve will be detected for $m$ decoy states is $1-(1-1/4)^m\approx1(m\geq20)$. For CNOT attacks~\cite{chun2005secure,gao2010cryptanalysis}, Eve can use the CNOT gate to eavesdrop user's operation key. Eve intercepts the qubit of the user $u$ as the control bit and her own qubit as the target so to perform the CNOT gate on them to build entanglement. After that, She sends the control qubit to the user $u$. When $u$ performs the operation on the qubit, Eve performs the CNOT gate again and measures her own qubit in the $Z$ basis to obtain the operation key of $u$. For example, when $u$'s decoy state is $|+\rangle$ ( $|-\rangle$ ) and Eve's qubit is $|0\rangle$, Eve may be detected based on the measurement result $|-\rangle$ ( $|+\rangle$ ). When $u$ uses $X$ basis to measure, the probability of this situation is $1/2$, as shown in Eq \eqref{eq-cnot-attack}. If the decoy state is $|0\rangle$ or $|1\rangle$, Eve cannot be detected as the decoy state cannot be entangled with Eve's qubit. Consequently, the probability that Eve can be discovered is $1-3/4^m\approx1$ for $m$ decoy states.
\begin{equation}\label{eq-cnot-attack}
\begin{aligned}
  {|+\rangle_u}{|0\rangle_e}&\stackrel{CNOT}\longrightarrow \frac{{{{\left( {\left| {00} \right\rangle  + \left| {11} \right\rangle } \right)}_{ue}}}}{{\sqrt 2 }} = \frac{{{{\left( {\left| { +  + } \right\rangle  + \left| { -  - } \right\rangle } \right)}_{ue}}}}{{\sqrt 2 }}\\
{\left|  -  \right\rangle _u}{\left| 0 \right\rangle _e}&\stackrel{CNOT}\longrightarrow \frac{{{{\left( {\left| {00} \right\rangle  - \left| {11} \right\rangle } \right)}_{ue}}}}{{\sqrt 2 }} = \frac{{{{\left( {\left| { +  - } \right\rangle  + \left| { - {\rm{ + }}} \right\rangle } \right)}_{ue}}}}{{\sqrt 2 }} 
\end{aligned}  
\end{equation}

For internal attacks, if the leader is malicious in our algorithms, the generated key can be controlled by him by changing the measurement results. Therefore, the participants must take turns as the leader to perform the Bell or GHZ measurements in the entangled sequence, which can ensure that the new key is not controlled by one. Therefore, the new key generation process is secure. 

In the distribution process, we need to demonstrate that the joining or leaving user $u$ remains unknown to the old keys in the key tree that are not associated with her, to ensure that she cannot obtain the previous or subsequent communications. In~\cite{campagna2015quantumsafe,pan2018satellite}, note that AES's cipher can adapt to the current quantum attacks by increasing its key size to rectify a vulnerability introduced by quantum computing. 
Therefore, AES-256 is as difficult for a quantum computer to break as AES-128 is for a classical computer. To keep the old keys private in the key tree, we use the old keys $k_i$ unknown to $u$ as seed keys of AES-256 to encrypt the new keys $k_i^{\prime}$ to complete the distribution. According to these rekey encryption messages, it is very difficult for the user $u$ to crack them, even if she has quantum computing power. Consequently, the group key distribution process is also secure.

Combining the security analysis of generation and distribution, we can conclude that Algorithms \ref{alg-qka-join} and \ref{alg-qka-leave} are secure.

\subsection{Quantum resource consumption}
To facilitate the analysis of the quantum resource consumpution in various QKA protocols, we give the following notation definitions. Let $N$ is the number of all participants (including users and servers) in a group session, and $\xi$ is the proportion of decoy states in each transmitted quantum sequence between $N$ group members, which is used to check the channel. $n$ is the length of the group key jointly generated by all participants.

In Shukla et al.'s multi-party QKA protocol~\cite{shukla2014protocols}, each participant (including users and servers) needs to prepare a Bell state and some decoy states, and then sends one particle of the Bell state and decoy states to the other participants. These participants perform $N$ transmissions and $N$ encryptions between group members as a round. Finally, the participant can extract a $1$-bit group key through the Bell measurement. For each participant, the number of qubits required to generate a $1$-bit group key is $2+\xi N$. 

In Sun et al.'s multi-party QKA protocol~\cite{sun2016efficient}, each participant prepares a $4$-qubit cluster state and some decoy states, and sends one particle of the cluster state and decoy states to others, so that they can perform a round of transmissions and encryptions between group members. She can extract a $2$-bit group key through the cluster measurement. For each participant, the number of qubits required to generate a $1$-bit group key is $2+\xi N/2$.

In Huang et al.'s multi-party QKA protocol~\cite{huang2016improved}, each participant prepares a single photons and $\xi$ decoy states, and sends to the other participants so to perform a round of transmissions and encryptions. She can deduce a $1$-bit group key based on single photon measurement, and the number of qubits required is $1+\xi N$.

In Zeng et al.'s multi-party QKA protocol~\cite{zeng2016multiparty}, the server needs to prepare a $N$-qubit GHZ state and $\xi(N-1)$ decoy states, and sends one particle of the GHZ state and $\xi$ decoy states to each of the other $N-1$ users. Each participant performs encryption operation on her own qubit, then return the qubit and $\xi$ decoy states to one user or server. Based on GHZ measurement, they can deduce a $1$-bit group key for $N$ participants, and the number of qubits required is $N+2\xi (N-1)$.

Based on the above QKA protocols, the quantum resource consumption to generate a $n$-bit group key is estimated as follows:
\begin{equation}\label{eq-cost-star}
\begin{aligned}
    &C_{bell}(N)=(2+\xi N)nN,\\
    &C_{cluster}(N)=(2+\xi N/2)nN,\\
    &C_{single}(N)=(1+\xi N)nN,\\
    &C_{ghz}(N)=(1+2\xi)nN-2\xi n.
\end{aligned}    
\end{equation}
Note that the consumption is $O(N)$ in QKA with GHZ states, and the others are $O(N^2)$, so we use GHZ states as quantum resource to generate group keys in our algorithms. An approximate measure of the quantum resource consumption is the number of qubits required by a join or leave request in a secure group.  

\begin{theorem}
    For a group of $N$ members, the average number of qubits required to generate a $n$-bit group key is $O(\log N)$ in our Algorithms \ref{alg-qka-join} and \ref{alg-qka-leave}.
\end{theorem}

\begin{proof}
For a key tree, $d$ and $h$ are the degree and height of the tree, respectively. If the key tree is used to instruct dynamic QGKA, we update only the keys of the $k$-node path of a joining or leaving node in the tree, as shown in Fig.\ref{fig-N-key-tree}. Note that a full and balanced key tree satisfies $N=d^{h-1}$, then $h-1=\log_d N$. Therefore, we need to update the $\log_d N$ keys at most. 

In Algorithms \ref{alg-qka-join} and \ref{alg-qka-leave}, note that the generation is used to update the keys of the associated path nodes in the key tree, which is the process of quantum resource consumption. The distribution is used to construct the rekey messages and distribute the new keys to group users, which is the process of classical information encrypted transmission. Therefore, we only need to estimate the number of keys updated in the generation process for a join and leave, i.e., $\log_d N$. Combined with Eq \eqref{eq-cost-star}, the join, leave, and average consumption of quantum resources in the key tree is given, respectively, as follows.

\begin{equation}\label{eq-tree-consumption}
\begin{aligned}
C_{join}(N)=&C_{ghz}(2)\log_d N\\
=&2(1+\xi)n\log_d N,\\
C_{leave}(N)=&C_{ghz}(d+1)(\log_d N-1)+C_{ghz}(d)\\
=&((1+2\xi)d+1)n\log_d N-(1+2\xi)n,\\
C_{avg}(N)=&\frac{C_{join}(N)+C_{leave}(N)}{2}\\
=&\frac{((1+2\xi)d+3+2\xi)n\log_d N-(1+2\xi)n}{2}\\
\sim& O(\log_d N).    
\end{aligned}
\end{equation}
For each time a user joins or leaves, the number of group members changes only slightly, so the degree $d$ of the key tree can be regarded as a constant before and after a user leaves.
As a result, the average number of qubits required to generate $1$-bit group key in our algorithms is $O(\log N)$. It means that the number of qubits required per join or leave only increases linearly with the logarithm of the group size. Dynamic QGKA based on the key tree can be scalable to large groups with frequent joins and leaves.
\end{proof} 

\section{Comparison results}
In this section, we discuss the degree of the key tree with group sizes and compare the consumption results with various quantum resources to demonstrate the advantages of the key tree in dynamic QGKA.

\subsection{Discussions on the degree of the tree key graph} 

The degree $d$ of the tree is the number of incoming edges in the key tree, i.e., the number of users in the subgroup. According to Eq \eqref{eq-tree-consumption}, note that it is directly related to the average consumption of quantum resources. Consequently, we need to find the optimal degree $d$ of the tree key graph for minimum consumption. 

\begin{figure}
    \centering
    \includegraphics[width=3in]{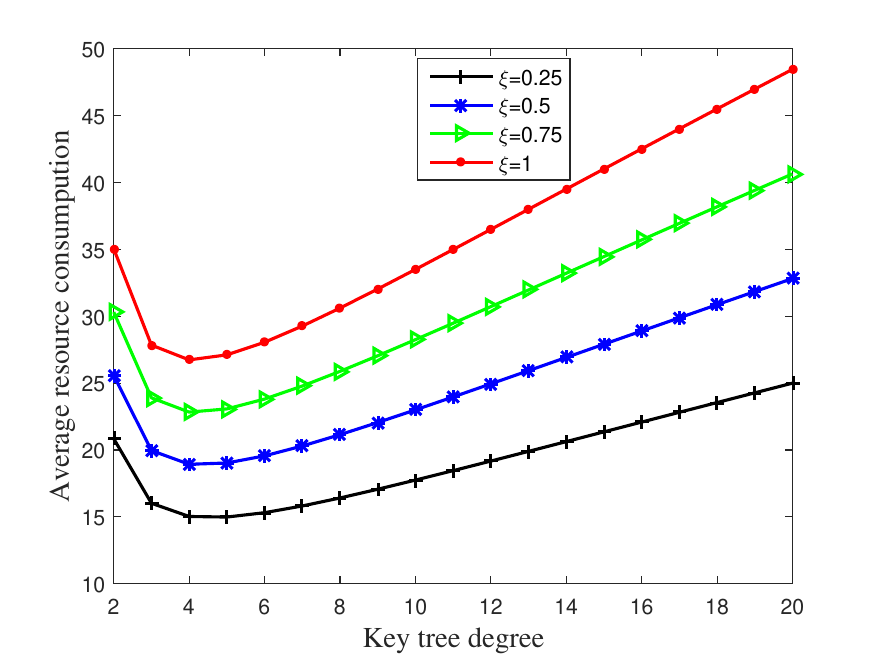}
    \caption{The average number of required qubits for a join and leave vs. the key tree degree. }
    \label{fig-optimal-degree}
\end{figure}

In Fig.\ref{fig-optimal-degree}, we show the relationship between the average number of required qubits and the tree degree when a user joins or leaves a session group. Black, blue, green, and red lines represent the results under the proportion of decoy states $\xi=0.25, 0.5, 0.75$, and $1$, respectively. From Fig.\ref{fig-optimal-degree}, note that when the degree is about $d=4$, the average number of required qubits reaches a minimum, and we can achieve the optimal degree of the key tree.

\subsection{Consumption comparison in the original QKA protocols} 
Based on Eq \eqref{eq-cost-star}, we can obtain the quantum resource consumption in the various original QKA protocols. To obtain the optimal QKA, we show the number of qubits required to generate the group key in these protocols.

Fig.\ref{fig-nontree-graph} shows the relationship between the number of required qubits and the number of users in four types of quantum resources. When we only use the original QKA protocols to distribute the group key, the number of qubits required is obtained based on Eq \eqref{eq-cost-star}. Black, blue, green, and red lines represent the results of QKA protocols with different quantum resources, which are Bell states~\cite{shukla2014protocols}, Cluster states~\cite{huang2016improved}, single photon states~\cite{sun2016efficient}, and GHZ states~\cite{zeng2016multiparty}, respectively. 

\begin{figure}
    \centering
    \includegraphics[width=3in]{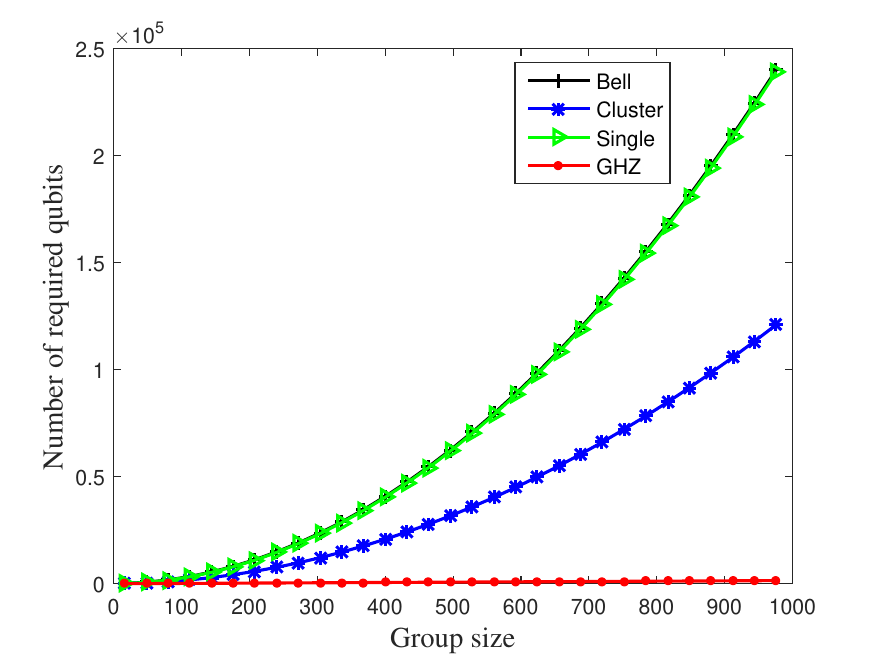}
    \caption{The number of required qubits in original QKA protocols vs. the number of users. }
    \label{fig-nontree-graph}
\end{figure}

From Fig.\ref{fig-nontree-graph}, note that the number of required qubits grows rapidly with group size $N$ when QKA protocols are used to distribute group keys based on the star key graph. Due to the large consumption of quantum resources, these protocols are difficult to apply directly to large-scale group communications. In the comparison of quantum resource consumption, the number of qubits required by QKA with GHZ state is much smaller than the other three. Since the GHZ state can be transmitted in a multicast manner, the consumption increases linearly with the group size. While the other three protocols are unicast, so more qubits are needed to agree on the group key. Consequently, QKA with GHZ state can provide better performance during group key generation.

\subsection{Consumption comparison for the frequent joins and leaves}
In Algorithms.\ref{alg-qka-join} and \ref{alg-qka-leave}, we introduce the key tree into QKA to address dynamic quantum group communication for user joins or leaves. To evaluate the performance of our algorithms, we show quantum resource consumption for the frequent joins or leaves in the tree key graph. Assuming that the number of users to join or leave satisfies the Poisson distribution, we perform the following simulations in dynamic groups.

\begin{figure}
    \centering
    \includegraphics[width=3in]{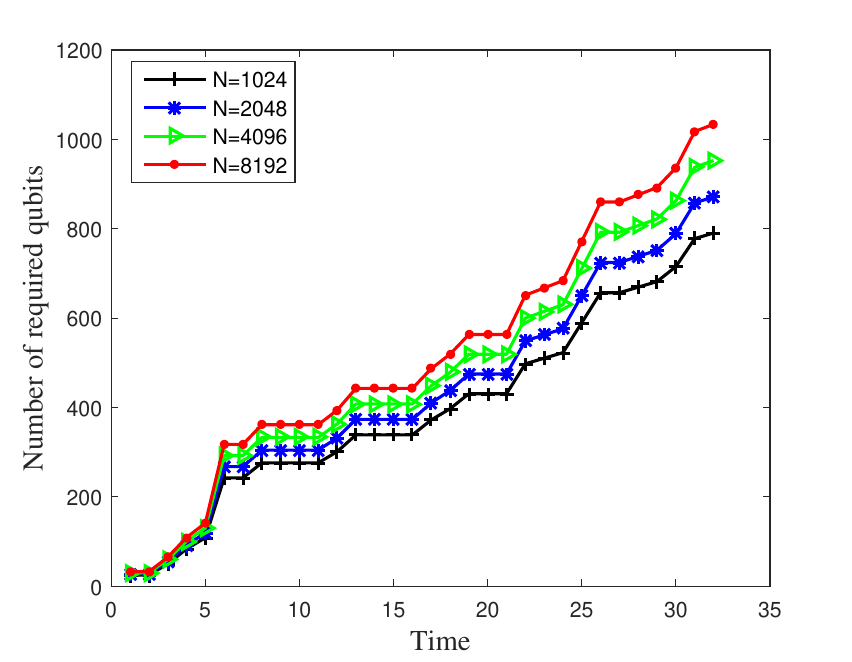}
    \caption{The number of required qubits for user joins and leaves vs. time in the key tree based on GHZ. Within a given time unit, the number of users who join or leave satisfies the Poisson distribution, i.e., \textit{Poisson}($\lambda=1$).}
    \label{fig-dynamic-ghz}
\end{figure}

\begin{figure}
    \centering
    \includegraphics[width=3in]{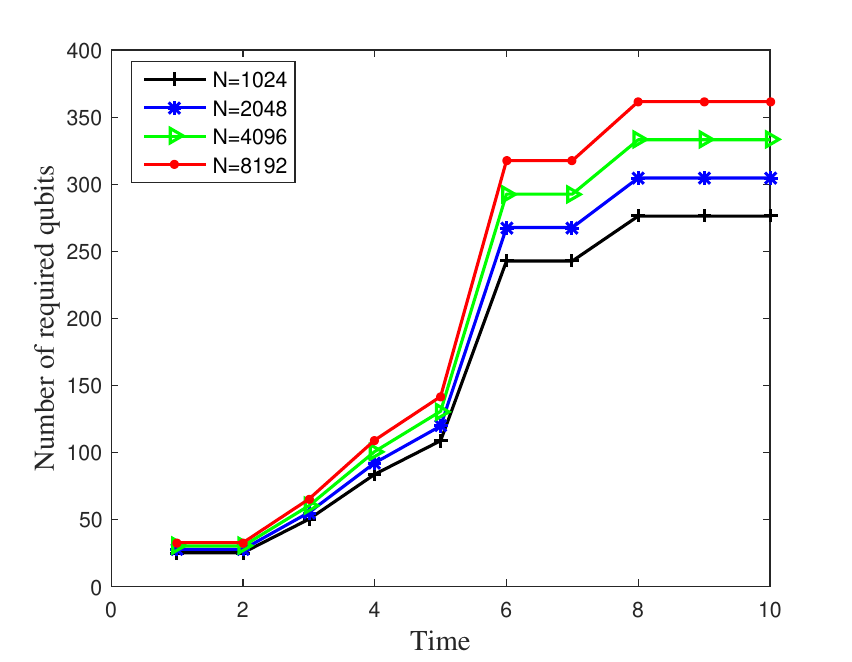}
    \caption{The number of required qubits for joins and leaves vs. time in the key tree based on GHZ. It is the partial enlargement of Fig.\ref{fig-dynamic-ghz}.}
    \label{fig-dynamic-ghz-magnify}
\end{figure}

\begin{figure}
    \centering
    \includegraphics[width=3in]{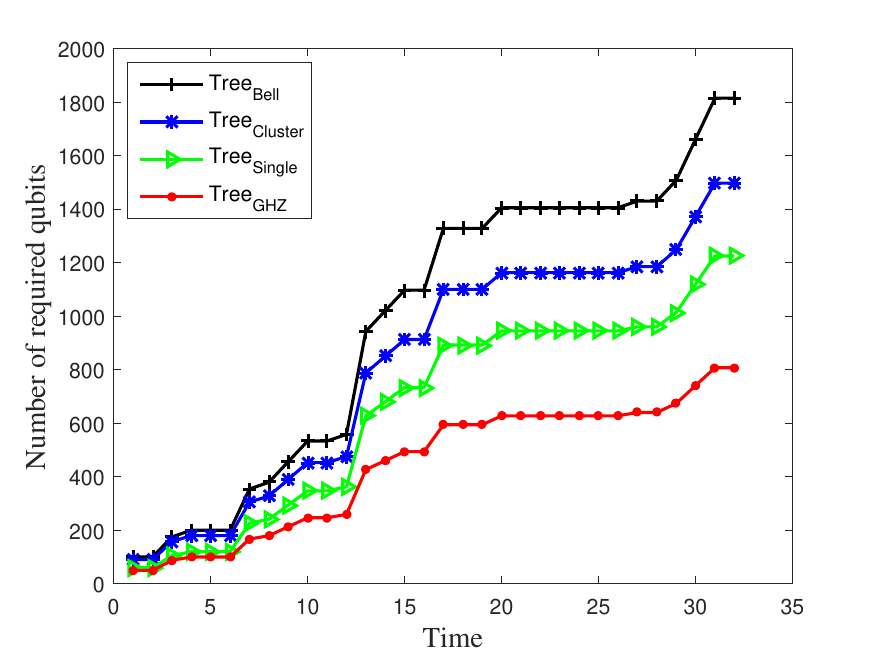}
    \caption{The number of qubits required for joins and leaves vs. time in the key tree based on various QKA protocols. The initial group size is $1024$. The number of users who join or leave during each time unit follows the Poisson distribution, i.e., \textit{Poisson}($\lambda=1$).}
    \label{fig-dynamic-tree-graph}
\end{figure}

\begin{figure}
    \centering
    \includegraphics[width=3in]{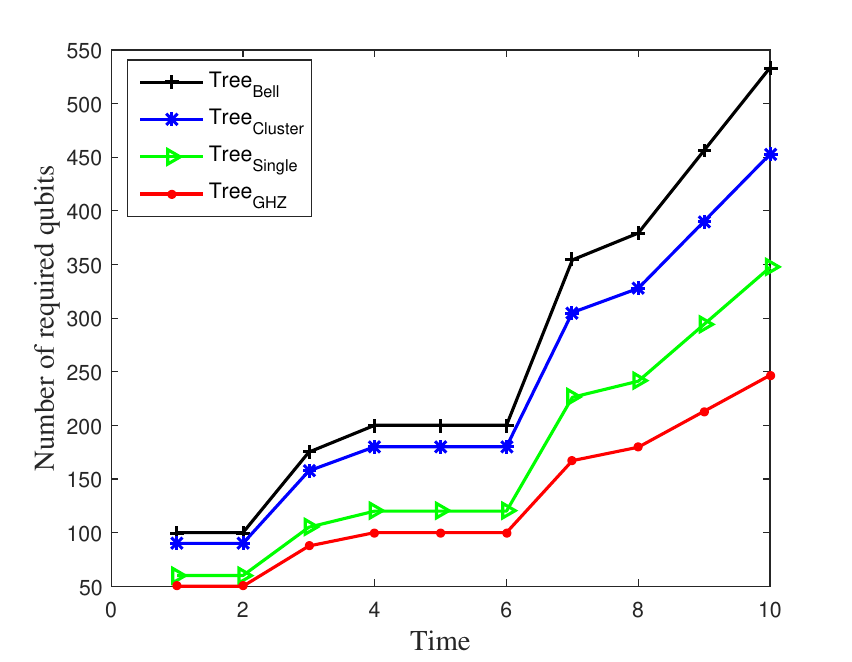}
    \caption{The number of qubits required for user joins and leaves vs. time in the key tree based on different quantum resources. It is the partial enlargement of Fig.\ref{fig-dynamic-tree-graph}.}
    \label{fig-dynamic-tree-graph-magnify}
\end{figure}


Fig.\ref{fig-dynamic-ghz} shows the cumulative number of qubits required changes over a given period of time as users join and leave when the tree key graph is applied to our algorithms. The black, blue, green, and red lines represent the group sizes of $N=1024$, $2048$, $N=4096$, and $8192$, respectively. In our discussion, the degree of the tree key graph $d$ is 4, the length of the group key $n$ is 1, and the proportion of decoy states $\xi$ in each quantum sequence is $0.25$. A unit time is a given period of time during which the number of users who join or leave satisfies the Poisson distribution ($\lambda=1$). From Fig.\ref{fig-dynamic-ghz}, note that the cumulative number of qubits required gradually increases as users join and leave over a period of time, and consumption grows with group size.  

Fig.\ref{fig-dynamic-tree-graph} shows the dynamic changes of the cumulative number of qubits required, qubits transmitted, and quantum gates used over a period of time as users join or leave the communication group, respectively. Our proposed Algorithms.\ref{alg-qka-join} and \ref{alg-qka-leave} use the tree key graph to instruct the GHZ-based QKA to update the group key for user joins or leaves, which is the red line. The black, blue, green, and red lines represent the simulation results of the key tree used in the Bell state~\cite{shukla2014protocols}, Cluster state~\cite{huang2016improved}, Single-photon state~\cite{sun2016efficient}, and GHZ state~\cite{zeng2016multiparty}-based QKA protocols, respectively. During each time unit, users can randomly join or leave the group, and their number follows the Poisson distribution, i.e., \textit{Poisson}($\lambda=1$). The initial size of the group is $1024$, and the other relevant parameters are the same as in Fig.\ref{fig-dynamic-tree-graph}. Compared to the original QKA protocols in Fig.\ref{fig-nontree-graph}, the group key update in Fig.\ref{fig-dynamic-tree-graph} requires a small number of qubits. This is because the key tree only needs to update a small number of hierarchical keys in each group key update. In addition, our algorithms (red line) show that the cumulative number of qubits required is less than that in other cases. This reason is that the GHZ-based QKA in the key tree can utilize a multicast manner to perform the sharing, transmission, and quantum gate operations on multi-qubit entangled states, respectively. In contrast, alternative scenarios need to execute these operations cyclically in a unicast approach to generate the group key. As a result, our algorithms can obtain better performance in dynamic quantum group communications. 

Our Algorithms \ref{alg-qka-join} and \ref{alg-qka-leave} demonstrate that the key tree can be used in GHZ-based QKA to reduce the quantum resource consumption for a join or leave. The key tree can yield superior improvement in dynamic quantum group communications, which is scalable to large groups with frequent joins and leaves. These findings highlight the potential of the key tree as the group key management strategy for scaling quantum group communications.

\section{Conclusions}
For large group sessions, to ensure that a user who joins or leaves cannot obtain the previous or subsequent communications, we propose two QGKA protocols for a join or leave in the key tree framework to instruct the server to update the group key, as shown in Algorithms \ref{alg-qka-join} and \ref{alg-qka-leave}. Compared with the original QKA protocols, our proposed Algorithms only require $O(\log N)$ qubits, which greatly reduces the consumption of quantum resources. This illustrates that the key tree-based hierarchy logical structure has great potential in dynamic quantum group communications. In future work, we will continue to conduct our exploration of the integration of key tree mechanisms into quantum key distribution protocols to enhance the secure group key rate in large-scale quantum group communications.


\bibliographystyle{IEEEtran}
\bibliography{reference}

\end{document}